\documentclass[a4paper,UKenglish,cleveref, autoref, thm-restate]{lipics-v2021}

\hideLIPIcs %uncomment to remove references to LIPIcs series (logo, DOI, ...), e.g. when preparing a pre-final version to be uploaded to arXiv or another public repository
\usepackage{amsmath}
\usepackage{amsbsy}
\usepackage{amssymb}
\usepackage{algorithm}

\usepackage{multicol}

\usepackage{graphicx}
\usepackage{cite}

\usepackage{tabularx}

\usepackage[table,xcdraw]{xcolor}

\usepackage{mathtools}
\usepackage{listings}
\usepackage[noend]{algpseudocode}

\newtheorem{lem}{Lemma}
\newtheorem*{relem}{Lemma}

\newtheorem{cor}{Corollary}
\newtheorem{property}{Property}

\algrenewcommand{\algorithmiccomment}[1]{\hfill// #1}
\newcommand{\myparagraph}[1]{\smallskip\noindent\textbf{#1}}

\newcommand{\new}[1]{\textbf{{#1}}}

\newcommand{\figb}[1]{\textcolor{black}{\textbf{{#1}}}}

\newcommand{\myskip}[1]{}

\newcolumntype{P}[1]{>{\centering\arraybackslash}p{#1}}

\newsavebox\CBox

\usepackage{booktabs}
\usepackage{multirow}
\usepackage{graphicx}

\usepackage{thmtools}

\usepackage{environ}

\newif\ifappendix
\newif\ifapprox

\NewEnviron{maybeappendix}[1]
    {\ifappendix
        \expandafter\global\expandafter\let\csname putmaybeappendix#1\endcsname\BODY%
    \else
        \expandafter\newcommand\csname putmaybeappendix#1\endcsname{}\BODY%
    \fi}
\newcommand{\putmaybeappendix}[1]{\csname putmaybeappendix#1\endcsname}

\NewEnviron{maybeapprox}[1]
    {\ifapprox
        \expandafter\global\expandafter\let\csname
        putmaybeapprox#1\endcsname\BODY%
    \else
        \expandafter\newcommand\csname putmaybeapprox#1\endcsname{}\BODY%
    \fi}
\newcommand{\putmaybeapprox}[1]{\csname putmaybeapprox#1\endcsname}
\appendixfalse %% False for full proofs
\approxfalse   %% False for full experiments

\title{Fast Kd-trees for the Kullback--Leibler Divergence and other Decomposable Bregman Divergences}

\titlerunning{Kd-trees for Decomposable Bregman Divergences}

\author{Tuyen Pham}{University of Florida, Gainesville, US}{tuyen.pham@ufl.edu}{}{}
\author{Hubert Wagner}{University of Florida, Gainesville, US}{hwagner@ufl.edu}{}{}

\authorrunning{T. Pham and H. Wagner} %TODO mandatory. First: Use abbreviated first/middle names. Second (only in severe cases): Use first author plus 'et al.'

\Copyright{Tuyen Pham, Hubert Wagner} %TODO mandatory, please use full first names. LIPIcs license is "CC-BY";  http://creativecommons.org/licenses/by/3.0/

 %TODO mandatory; please add comma-separated list of keywords

\category{} %optional, e.g. invited paper

\relatedversion{} %optional, e.g. full version hosted on arXiv, HAL, or other respository/website
\supplement{}%optional, e.g. related research data, source code, ... hosted on a repository like zenodo, figshare, GitHub, ...
\funding{2022 Google Research Scholar Award in Algorithms and Optimization.}%optional, to capture a funding statement, which applies to all authors. Please enter author specific funding statements as fifth argument of the \author macro.

\acknowledgements{}%optional

\nolinenumbers %uncomment to disable line numbering

\EventEditors{Xavier Goaoc and Michael Kerber}
\EventNoEds{2}
\EventLongTitle{The 39th International Symposium on Computational Geometry}
\EventShortTitle{SoCG 2023}
\EventAcronym{SoCG}
\EventYear{2022}
\EventLogo{figs/socg-logo}
\SeriesVolume{224}
\ArticleNo{63}
\begin{document}

\maketitle

\begin{abstract}
The contributions of the paper span theoretical and implementational results. First, we prove that Kd-trees can be extended to spaces in which the distance is measured with an arbitrary Bregman divergence. Perhaps surprisingly, this shows that the triangle inequality is not necessary for correct pruning in Kd-trees. Second, we offer an efficient algorithm and C++ implementation for nearest neighbour search for decomposable Bregman divergences.

The implementation supports the Kullback--Leibler divergence (relative entropy) which is a popular distance between probability vectors and is commonly used in statistics and machine learning. This is a step toward broadening the usage of computational geometry algorithms.

Our benchmarks show that our implementation efficiently handles both exact and approximate nearest neighbour queries. Compared to a naive approach, we achieve two orders of magnitude speedup for practical scenarios in dimension up to 100. Our solution is simpler and more efficient than competing methods.
\end{abstract}

\section{Motivation}
Nearest neighbour search is a fundamental method offered by computational geometry, with applications in a wide range of fields. Bentley's k-dimensional tree, Kd-tree for short, is among the simplest and most practical data structures for this task. 

%Like many other computational geometry techniques, Kd-trees were designed for the Euclidean space, and later extended to more general metric spaces. However, many modern geometric problems --- particularly in widely construed data science --- rely on distances that are not proper metrics. For example, it is common to represent data as a collection of probability vectors measured with specialized dissimilarities. Consequently, standard geometric algorithms do not work with non-metric distances --- but they can often be extended to such settings. 
Like many other computational geometry techniques, Kd-trees were initially designed for Euclidean space and later extended to more general metric spaces. However, many modern geometric problems, particularly in data science, use distances that are not proper metrics. For example, it is common to represent data as probability vectors and use specialized dissimilarities to measure distances between them. While standard geometric algorithms do not work with non-metric distances, they can often be extended to such settings.

Indeed, it is interesting that many computational geometry algorithms  -- that are typically assumed to require a metric -- can work with significantly weaker assumptions. We will mention prominent examples momentarily, many of which originated at SoCG, and prove that the above statement extends to Kd-trees. In particular, correctness and efficiency guarantees can be retained under much weaker assumptions than typically assumed. Specifically, for dissimilarity measures that are asymmetric and do not fulfill the triangle inequality.

%\myparagraph{A brief detour into applications.} One practical example of such a non-metric distance is the Kullback--Leibler (KL) divergence~\cite{kullback1951information}. Inspired by information theory~\cite{shannon1948mathematical}, it is the standard way of comparing discrete probability distributions. In practice it is, for example, used as a loss function minimized in machine learning~\cite{van2008visualizing, mcinnes2018umap, murphy2012machine}. Answering approximate nearest neighbours queries is becoming an important component of modern machine learning. In particular, retrieval-augmented generation (RAG) is an attempt to make large language models more robust by searching for existing documents supporting generated answers. This is done by a nearest neighbour search. Heuristic methods that lack performance guarantees are typically used --- recently Indyk and Xu~\cite{indyk23} warned that the methods used in such contexts can catastrophically fail.  Regardless, there is a growing field of \emph{vector databases} focusing on supporting nearest neighbour queries~\cite{han2023comprehensive}. Most major databases are starting to support such geometric operations.
\myparagraph{Applications.} One practical example of such a non-metric distance is the Kullback--Leibler (KL) divergence~\cite{kullback1951information}. Originating in information theory~\cite{shannon1948mathematical}, the KL divergence is a standard way of comparing discrete probability distributions (probability vectors). For example, it is used as a loss function minimized in machine learning~\cite{van2008visualizing, mcinnes2018umap, murphy2012machine}. Approximate nearest neighbours queries are becoming an increasingly important component of modern machine learning. In particular, retrieval-augmented generation (RAG) aims to improve large language models by searching for existing documents to support generated answers. This is done by a nearest neighbour search within probability vectors. Typically, heuristic methods that lack performance guarantees are used. Recently Indyk and Xu~\cite{indyk23} warned that the methods used in such contexts can catastrophically fail. Regardless, there is a growing field of \emph{vector databases} focusing on supporting nearest neighbour queries~\cite{han2023comprehensive}.

\myparagraph{Problem statement.}
%Guided by curiosity and the above practical considerations, we revisit the topic of nearest neighbour search. We focus on algorithms that offer exact answers, as well as approximate search with guarantees. We are interested in non-metric distance measurements, in particular geometries induced by the so-called Bregman divergences --- primarily because the KL divergence turns out to be a member of this family. 
We revisit the topic of nearest neighbour search, focusing on algorithms that provide exact answers as well as approximate results with guarantees. In particular, we investigate non-metric geometries induced by Bregman divergences --- of which the KL divergence is a prominent member.

Given a finite collection of points $X \subset \Omega \subset \mathbb{R}^d$ and a query point $q \in \Omega$, select $k$ points from $X$ with the smallest distance to $q$. Specifically, the distance will be measured using a Bregman divergence, which we discuss in Section~\ref{sec:breg}. We will design and implement a modified Kd-tree for answering queries in this setting.

\myparagraph{Contributions.} We list the contributions of this paper:
\begin{enumerate}
    \item Theoretical results on correctness and efficiency of Kd-tree queries in the setting of Bregman divergences.
    \item The first implementation of Kd-trees for Bregman divergences. It is currently the fastest method for exact k-nearest neighbour queries in the Bregman setting and works for arbitrary decomposable divergences\footnote{We are working on incorporating our implementation into the popular scikit-learn~\cite{scikit-learn} library, which will increase both the generality and efficiency of its existing Kd-trees implementation.}.
    \item Benchmarks showing the method is usable in practical situations.
%    \item Motivate practical uses of Bregman divergences, focusing on the KL divergence.
\end{enumerate}
%In particular, the last two points go beyond Kd-trees. We hope that the research on Bregman divergences can be reinvigorated in the computational geometry community and that robust geometric tools could be more widely used to solve practical problems.

%\myparagraph{Structure of the paper.}

\section{Related work} \label{sec:related}
Kd-Trees were introduced by Jon Bentley in 1975~\cite{Bentley1975MultidimensionalBS}. Further improvements were made by him, Friedman and Finkel~\cite{kd2_bentley} and many others. Bregman divergences were introduced by Lev Bregman in 1967~\cite{BREGMAN1967200}.

\ifapprox
    Many computational geometry techniques have been extended to operate with Bregman divergences instead of a metric. In the context of nearest neighbour search, Cayton blazed the trail by extending~\cite{CaytonBBTrees} ball-trees and implementing prototype software for the KL and IS divergences. He also proved theoretical results towards extending Kd-trees~\cite{caytonphd2009}, which we strengthen as well as provide algorithms and an efficient implementation. In turn, Nielsen, Piro, and Barlaud extended Vantage point trees~\cite{5202635, VPTree}. The same authors further improved Bregman ball-trees~\cite{Bregman_Ball_Trees, BBTreepp}. R-trees and VA-files were extended by Zhang and collaborators~\cite{RTreesVAFiles}, inspiring Song and collaborators~\cite{song2020brepartition}. Many approximate search methods have been adapted to the Bregman setting. We discuss related works and present experimental comparisons in Appendix~\ref{sec:approxDisc}.
\else
    Many computational geometry techniques have been extended to operate with Bregman divergences instead of a metric. In the context of nearest neighbour search, Cayton blazed the trail by extending~\cite{CaytonBBTrees} ball-trees and implementing prototype software for the KL and IS divergences. He also proved theoretical results towards extending Kd-trees~\cite{caytonphd2009}, which we strengthen as well as provide algorithms and an efficient implementation. In turn, Nielsen, Piro, and Barlaud extended Vantage point trees~\cite{5202635, VPTree}. The same authors further improved Bregman ball-trees~\cite{Bregman_Ball_Trees, BBTreeGithub}. R-trees and VA-files were extended by Zhang and collaborators~\cite{RTreesVAFiles}, inspiring Song and collaborators~\cite{song2020brepartition}. Ring-trees combined with a quad-tree decomposition have been proven to work sublinearly for finding approximate nearest neighbours by Abdullah, Moeller, and Venkatasubramanian~\cite{Bregman_ring_tree}. In 2013, Boytsov and Naidan developed their own Bregman VP-trees extension~\cite{BoytsovNaidan_VPTrees} for approximate nearest neighbours. Naidan later incorporated his VP-trees and Cayton's ball-trees into the Non-Metric Space Library (NMSLIB)\cite{nmslib}. This library also includes other approximate Bregman similarity searches including small world graphs~\cite{MalkovPonomarenkoLogvinovKrylov_SWG}. The hierarchical navigable small world graph has been a popular choice for similarity searches in vector databases~\cite{HNSW_Zilliz, HNSW_MariaDB, HNSW_MongoDB} and perform well in benchmarks for metrics~\cite{ANN_Benchmarks}. However, its implementation in NMSLIB is currently experimental for Bregman divergences~\cite{HNSW_Git}.
    Recently Abdelkader, Arya, da Fonseca and Mount proposed an approach to proximity search in non-metric settings, which includes Bregman divergences~\cite{abdelkader2019approximate}; as we understand it, this has not yet been implemented.
\fi

More broadly, Banerjee and collaborators extended $k$-means clustering~\cite{JMLR:v6:banerjee05b} to arbitrary Bregman divergences -- with the surprising twist that the existing algorithm works without changes. Coresets have also been extended to the Bregman setting by Ackermann and Bl{\"o}mer~\cite{Ackermann_Blomer_k_median}. Nielsen, Boissonnat, and Nock developed Bregman Voronoi diagrams and Delaunay triangulations~\cite{Bregman_Voronoi}. Edelsbrunner and Wagner extended topological data analysis methods to the Bregman case~\cite{EdWa16}. 

% This should be here. The software I mention below also includes the below things.
%On a slightly different note, we mention some developments related to geometries on the standard simplex coming from the computational geometry community. One example is the work of Nielsen and Shao, considering Hilbert geometries on the simplex~\cite{HSGball-2017}. Kyng, Phillips and Venkatasubramanian showed Johnson--Lindenstrauss-type results on the simplex~\cite{kyng2010johnson, kyng2011dimensionality}.

In the Euclidean case, robust software is available for all of these techniques.
One popular package in the Euclidean case is the ANN library by Mount and Arya~\cite{ANN_kd, ANN_Manual, ANN_boundary, ANN_optimalAlgorithm}. Our current implementation is inspired by this library. % which we will compare with the Bregman ball-tree implementation by Cayton~\cite{BBTreeGithub} and heuristic methods in NMSLIB~\cite{nmslib}.

\section{Background on Bregman divergences}
\label{sec:breg}
We begin by setting up definitions for Bregman divergences~\cite{BREGMAN1967200}, which we use as a measure of distance. These divergences are usually asymmetric and do not generally satisfy the triangle inequality -- and as such do not define a proper metric. Despite this limitation, decomposable Bregman divergences will efficiently work with Kd-trees with minimal changes. 

Each Bregman divergence is parametrized by a convex function with particular properties~\cite{bauschke1997legendre}. We set the stage by letting $\Omega\subseteq\mathbb{R}^d$ be an open convex set. Next, we define a \new{function of Legendre type}~\cite{Rockafellar+1970} as a function $F: \Omega \to \mathbb{R}$ that is: \nopagebreak
\begin{enumerate}[I] \nopagebreak
    \item differentiable and \nopagebreak
    \item strictly convex. \nopagebreak
    \item We additionally require that $\lim\limits_{x\to \partial\,\Omega}\|\nabla F(x)\|=\infty$, provided $\partial\,\Omega$ is nonempty. \nopagebreak
\end{enumerate}
The third requirement is often omitted, but will prove important in Section~\ref{sec:correctnessofpruning}.

Given a function $F$ of Legendre type, the \new{Bregman divergence generated} by $F$ is a function $D_F : \Omega \times \Omega \to \mathbb{R}$. The value of the divergence between $x$ and $y$ is the difference between $F(x)$ and the best affine approximation of $F$ at $y$ also evaluated at $x$, or simply
\begin{align}
    D_F(x\|y)=F(x)-(F(y)+\langle\nabla F(y), x-y\rangle).
\end{align}
See Figure~\ref{fig:bregman_div} for an illustration. We refer to $D_F(x\|y)$ as the divergence in the \new{direction} from $x$ to $y$. Due to the lack of symmetry, we will be mindful about the direction in which we compute it.

All Bregman divergences fulfill the following
\begin{property}[Bregman Nonnegativity]\label{positivity}
$D_F(x\|y) \ge 0$ for each $x,y \in \Omega$, with equality if and only if $x=y$.
\end{property}

Despite failing to satisfy the requirements for a metric, various computational geometry algorithms extend to Bregman divergences. Also, despite the seemingly simple definition, the resulting divergences have interesting properties and interpretations.% Specifically, in~\cref{sec:KL_interpretation} we look into an interpretation and usage of the KL divergence.

\begin{figure}
    \centering
    \includegraphics[width = .4\textwidth]{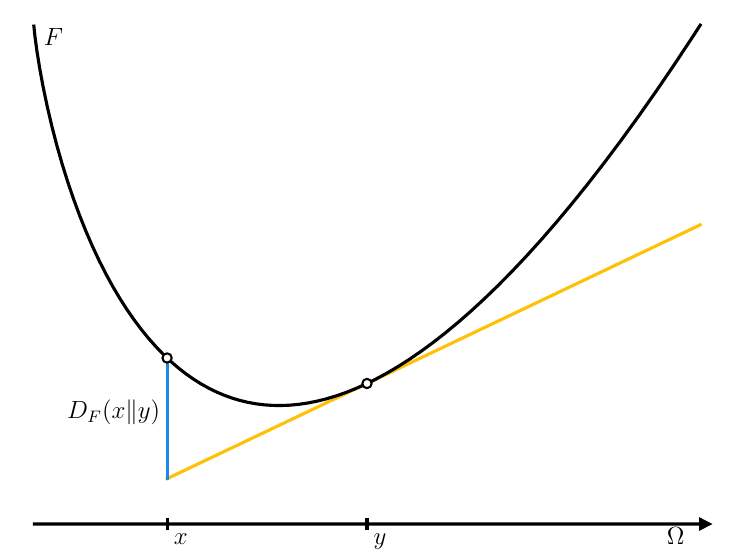}
    \caption{Visualization of a Bregman divergence construction for a one-dimensional domain.}
    \label{fig:bregman_div}
\end{figure}

\myparagraph{Decomposable Bregman divergences.} Our focus is on a sub-family of Bregman divergences called \new{decomposable Bregman divergence}~\cite{Decomp_def1, Decomp_def2}. They are generated by a function $F = \sum_{i=1}^d f_i$, where each $f_i$ is a univariate function of Legendre type. The function, $F$, generates a Bregman divergence of the form $D_{F}(x\|y)=\sum_{i=1}^{d}D_{f_i}(x_i\|y_i)$ for $x_i,y_i$ lying in the domain of $f_i$~\cite{DBLP:journals/corr/abs-1810-09113}. Most divergences used in practice belong to this family.

\begin{table}
    \centering
    \caption{List of common decomposable Bregman divergences.}
    \begin{tabular}{@{}clll@{}}
    \toprule
    Domain & Generator      & Divergence     & Name                          \\ \midrule
    $\mathbb{R}^d$    & $\sum_{i=1}^d x_i^2$          & $\sum_{i=1}^d(x_i-y_i)^2$                                       & Squared Euclidean                   \\
    $\mathbb{R}_+^d$  & $-\sum_{i=1}^dx_i\log_{2}\frac{1}{x_i}$  & $\sum_{i=1}^dx_i\log_{2}\frac{x_i}{y_i} + \frac{y_i-x_i}{\ln2}$ & Generalized Kullback--Leibler (GKL) \\
    $\triangle^{d-1}$ & $-\sum_{i=1}^dx_i\log_2\frac{1}{x_i}$    & $\sum_{i=1}^dx_i\log_2\frac{x_i}{y_i}$                          & Kullback--Leibler (KL)              \\
    $\mathbb{R}_+^d$  & $-\sum_{i=1}^d\log x_i$       & $\sum_{i=1}^d\frac{x_i}{y_i}-\log\frac{x_i}{y_i} - 1$           & Itakura--Saito (IS)                 \\
    $\mathbb{R}^d_+$        & $-\sum_{i=1}^d-\sqrt{x}$ & $\sum_{i=1}^d \frac{\sqrt{y_i}}{2} + \frac{x_i}{2\sqrt{y_i}}- \sqrt{x_i}$  & Bhattacharyya-Like                      \\ 
    $\Omega_1\cap\Omega_2$& $\lambda F_1+(1-\lambda)F_2$&$\lambda D_{F_1}+(1-\lambda)D_{F_2}$& Hybrid divergence\\\bottomrule
    \end{tabular}
    \label{tab:Decomp_Bregmans}
\end{table}
We list common decomposable Bregman divergences in Table~\ref{tab:Decomp_Bregmans}. The most commonly used decomposable Bregman divergence is the \new{squared Euclidean distance}. Of particular interest is the \new{generalized Kullback--Leibler divergence} (GKL) defined $\mathbb{R}_{+}^d$ and generated by the negative Shannon entropy. It reduces to the standard KL divergence for points on the \new{open standard simplex}, $\triangle^{d-1}=\{x\in\mathbb{R}_+^d:\sum_{i=1}^{d}x_i = 1\}$.  Another example is the \new{Itakura--Saito} (IS) divergence~\cite{itakura1968analysis}, generated by the negative Burg's entropy. It is useful for working with speech and sound data~\cite{do2002wavelet}. There are many other decomposable Bregman divergences, some inspired by popular tool in statistics such as the Bhattacharyya distance. Additionally, given functions of Legendre type, $F_1$ and $F_2$, we can form a hybrid divergence, which can be viewed as an interpolation between the two Bregman divergences.

One outlier is the squared Mahalanobis distance~\cite{mahalanobis1936generalised} which is popular in statistics. While not decomposable, nearest neighbour problems involving this divergence can be reduced to the squared Euclidean distance that is decomposable. Another one is the KL divergence on the \emph{closed simplex}. While it does not fall under our definition, it can be treated as a limit case of the KL on the open simplex. 

Overall, restricting our attention to decomposable divergences over open domains is not going to limit the choice of divergences handled in practice.

\myparagraph{Bregman balls.} Due to the asymmetry, one can define two types of Bregman balls~\cite{BregmanBallDef}. We start from the \new{primal Bregman ball} of radius $r\geq 0$ centered at $q$ which is defined as 
\begin{align}
    B_{F}(q;r)=\{y\in \Omega: D_F(q\|y)\leq r\}.
\end{align} 
Namely, it is the collection of points with Bregman divergence measured \emph{from} the center not exceeding $r$. See~\cref{fig:GKL_IS_Open_Ball} for an illustration.
%Primal Bregman balls have a particularly nice geometric interpretation: given a light source at point $(q, F(q)-r)$, the primal ball $B_F(q;r)$ is the \emph{illuminated} part of the graph of $F$ projected vertically onto $\Omega$.

%In contrast to Euclidean balls, Bregman balls have a less uniform geometry, and in some cases may even be nonconvex.
\begin{figure}
    \centering
    \includegraphics[scale=0.35]{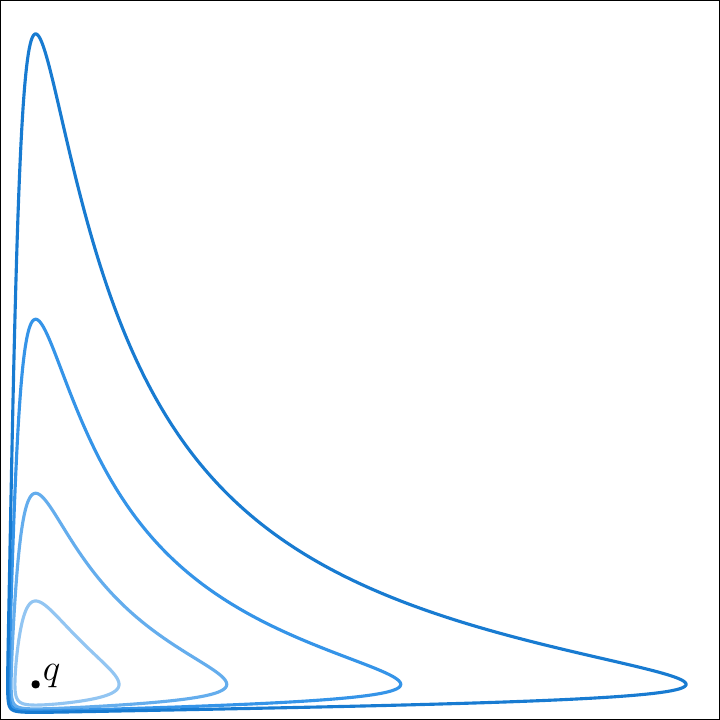}\hfil
    \includegraphics[scale=0.35]{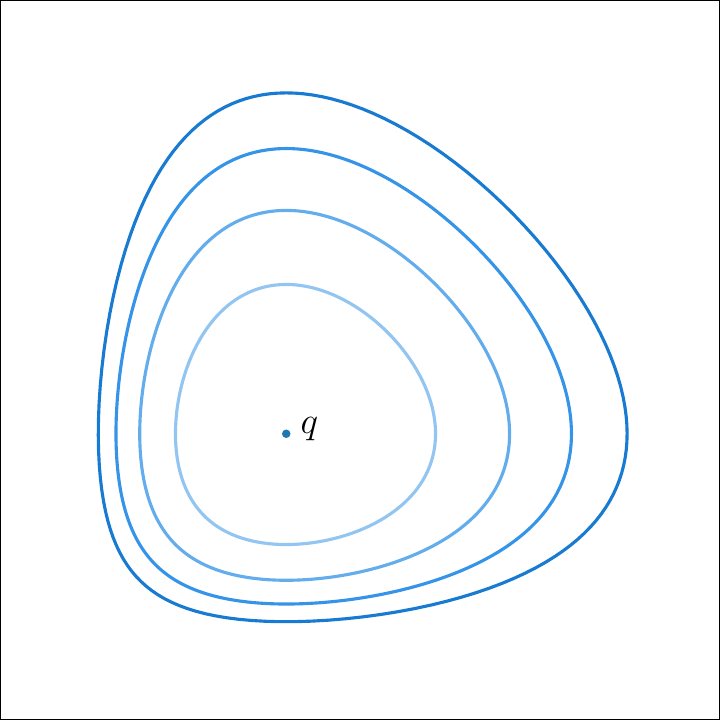}
    \caption{Left: primal Itakura--Saito balls. Right: primal generalized Kullback--Leibler balls.}
    \label{fig:GKL_IS_Open_Ball}
\end{figure}

The \new{dual Bregman ball} of radius $r\geq0$ centered at $q$ is defined as
\begin{align}
    B_F'(q;r) = \{y\in\Omega: D_F(y\|q)\leq r\}.
\end{align}
%The dual ball has a similarly nice geometric interpretation. Given the tangent plane at $(q, F(q))$, we shift it up by $r$, chopping off the upper portion of the graph of $F$. The dual Bregman ball $B'_F(q;r)$ is the remaining part projected vertically onto $\Omega$.% Both geometric interpretations are illustrated in Figure~\ref{fig:PrimalDualBall_def}.
%\begin{figure}
%    \centering
%    \includegraphics[width = .49\textwidth]{figs/1D_Primal_ball.pdf}
%    \includegraphics[width = .49\textwidth]{figs/1D_Dual_ball.pdf}
%    \caption{Geometric interpretation of a primal Bregman ball (left) and the dual Bregman ball (right) in dimension 1}
%    \label{fig:PrimalDualBall_def}
%\end{figure}

%\begin{figure}
%    \centering
%    \includegraphics[width=0.8\textwidth]{figs/disconnected.png}
%    \caption{The blue, partially obstructed, object is a primal Itakura--Saito ball at $(0.2, 0.6, 0.2)$ with radius $2.0$. Its restriction to the brighter hyperplane has two connected components,  outlined in thick red, separated by the other hyperplane. Is this restriction a Bregman ball? \\ If Bregman balls displayed similar pathological behaviours, Kd-tree may incorrectly prune nearest neighbour candidates. A significant portion of the paper aims to keep such cases at bay.}
%    \label{fig:bad}
%\end{figure}

As seen in Figure~\ref{fig:GKL_IS_Open_Ball}, and observed in~\cite{Bregman_Voronoi}, primal Bregman balls can be non-convex (when viewed as a subset of Euclidean space).  It is reasonable to question if all balls are necessarily connected. %In Figure~\ref{fig:bad} we hint at situations in which this could potentially occur.
In Section~\ref{sec:correctnessofpruning} we will show that the balls are indeed connected, and emphasize why this property is crucial. 

\myparagraph{Bregman projections.}
While different in many aspects from metrics, Bregman divergences often exhibit familiar behaviours. We mention standard results related to projections, which we sharpen in the subsequent sections.

Given a Bregman divergence $D_F$, we consider the \new{Bregman projection} to a point $q$ from a nonempty $C \subset \Omega$: 
\begin{align*}
    \operatorname{proj}_F(q, C) &= \operatornamewithlimits{arginf}_{x\in C}D_{F}(x\|q).
\end{align*}
When $C$ is closed and convex, this projection exists and is unique. In this case, we declare this point to be the Bregman projection of $q$ onto $C$. In analogy with projection distance, we define the Bregman \new{projection divergence} as $D_F(\text{proj}_F(q, C)\|q)$, the infimum of divergences from $C$ to $q$. We state the following useful statement~\cite{Bregman_Voronoi}. %First, we fix notation. Referring to an affine subspace in $\Omega$, we mean affine in $\mathbb{R}^d$ and consider its restriction to $\Omega$.

\begin{lem}[Bregman Projection~\cite{Bregman_Voronoi}] \label{lem:proj}
    Given a nonempty closed convex set $C\subset\Omega$ and $q\in \Omega$, denote $q_C = \operatorname{proj}_F(q, C)$. For all $x\in C$:
    \begin{align*}
        D_F(x\|q)\geq D_{F}(x\|q_C) + D_{F}(q_C\|q).
    \end{align*}
    If $C$ is an affine subspace, the above is an equality.
\end{lem}
%The inequality can be interpreted as a \emph{reverse triangle inequality}, and the equality as a \emph{Bregman Pythagorean theorem}. Both appear particularly useful in our setting, since the Kd-tree construction partitions the domain into convex bodies. 

However, this theorem only applies when $D_F$ is computed to $q$ from $q_C$. By restricting the setup to axis-aligned boxes, we obtain a similar result working in both directions. We provide results for divergences computed \textit{from} a query
\ifappendix
with proofs found in Appendix~\ref{sec:proofs}
\fi, but results and proofs for the reverse direction are analogous.

\section{Kd-trees}
We briefly overview a version of the Kd-tree data structure introduced by Bentley~\cite{Bentley1975MultidimensionalBS}, focusing on nearest neighbour queries. It is a binary tree which encodes recursive partitioning of $\mathbb{R}^d$ (in practice: a sufficiently large box contained in it) into axis-aligned boxes. We consider a variant in which each node corresponds to an axis-aligned box and the data points are stored only in the leaves. We highlight the changes required to extend the standard method to the Bregman setting.

%\begin{figure}
%    \centering
%    \includegraphics[width=0.4\textwidth]{figs/kd_tree.pdf}
%    \includegraphics[width=0.4\textwidth]{figs/planes_with_ISEuc.pdf}
%    \caption{\figb{Left:} Subdivision of 2-dimensional space performed the Kd-tree construction.
%    \figb{Right:} Pruning test as an intersection of a Bregman ball with a box. A primal Itakura--Saito ball and a Euclidean ball intersect different boxes, leading to different pruning decisions.}    
%    \label{fig:kd}
%\end{figure}

\myparagraph{Construction.}
Kd-trees partition the space by cutting it with axis-aligned hyperplanes, often called splitting planes or cutting planes. %; see Figure~\ref{fig:kd}.
The details of the construction, i.e. the order and locations of the splits, can have a significant impact on efficiency of Kd-trees~\cite{Bentley1975MultidimensionalBS, kd3_sproull, kd3_bentley, kd2_bentley}. However, the construction does not depend on the choice of a metric, or divergence. We therefore only mention that the standard splitting methods work in the Bregman case. It is worth emphasizing that once the tree is constructed, each query can be made using any decomposable divergence.

\myparagraph{Nearest neighbour queries in the Bregman case.}
Consider a Kd-tree constructed from a finite set of \new{data points}, $X \subset \Omega \subset \mathbb{R}^d$. To simplify exposition, we focus on finding the single, exact nearest neighbour of the query point $q$ among the points in $X$. More precisely, we consider the primal Bregman nearest neighbour, $\operatornamewithlimits{argmin}_{x \in X} D_F(q\|x)$. The proposed implementation is general.

We recall that the query can be performed using a simple recursive procedure. It traverses the tree trying to prune as many subtrees as possible, while guaranteeing that all viable candidates are considered. We overview the algorithm in the Bregman case now, and present an implementation in Section~\ref{sec:implementation}.

\begin{itemize}
\item The base case: divergences from $q$ to the points stored in the leaf are compared with the divergence to the current best candidate, $r_{nn}$, which is updated if needed. 

\item We first visit child nodes based on the relative location of the query point and the splitting plane at this node. As such, this step does not depend on the choice of divergence.
%We first visit the child which is more likely to contain points closer to $q$. The node is chosen based on the relative location of the query point and the splitting plane at this node. As such, this step does not depend on the choice of distance or divergence.

\item Moving back to the root, we visit the remaining subtrees only if they cannot be safely pruned. This is decided by what we call the \new{pruning test}. Conceptually, we rephrase it as an \new{intersection test} between the axis-aligned box corresponding to the remaining node and a primal Bregman ball. Specifically, we mean the ball centered at $q$ of radius $r_{nn}$. %See Figure~\ref{fig:kd} for an illustration. 
Clearly, if the box and the ball are disjoint, all data points in the box are in the complement of the ball, and are therefore too far away to contribute. 
\end{itemize}

\myparagraph{Details.} We mention that finding the dual Bregman nearest neighbour is completely analogous. Finding the $k$ nearest neighbours, is another easy modification involving a priority queue. In this case $r$ is set to the divergence to the current $k$-th nearest neighbour, or infinity. To allow approximate queries, the radius of the ball is decreased to $\frac{r}{1+\varepsilon}$. Our implementation supports all these options. We skip the details for brevity.

\myparagraph{Pruning test in practice.} In practice, to determine if a given box can be safely pruned, we will perform a Bregman projection of the query point onto the boundary of the box. Before we describe the implementation, we must prove that --- despite the lack of symmetry and triangle inequality --- correct and efficient pruning is possible. To this end, we focus on problems related to intersecting a Bregman ball $B \subset \Omega \subset \mathbb{R}^d$ with an axis-aligned box $A \subset \mathbb{R}^d$.% In several steps, we will simplify the setup and focus on intersections with axis-aligned hyperplanes, and then on performing Bregman projections. 
% Correctness
%     projecting onto hyperplane also checks half space intersection iff balls are connected
% Computations
%     box intersection <- box projection <- many hyperplanes projection <-(`decomposability')- single hyperplane projection

We divide our argument into two parts. Part I is presented in Section~\ref{sec:correctnessofpruning} and is more topological: we show that intersecting $B$ with the boundary of $A$ is an equivalent test. This argument works for arbitrary Bregman divergences, not only decomposable ones. Part II is presented in Section~\ref{sec:eff} and is more geometric: we replace the intersection test with a projection and show it can be computed in a simple efficient way. This part is specific to decomposable divergences.

\section{Proof of pruning correctness}\label{sec:correctnessofpruning}
We consider a Bregman ball $B \subset \Omega \subset \mathbb{R}^d$ and an axis-aligned box $A \subset \mathbb{R}^d$. The intersection test checks if the intersection $A$ and $B$ is nonempty. We first prove a crucial result.
%Our first step is to simplify this test.

%\subsection{Boundary intersection test}
%We can assume that the center of the Bregman ball $B$ lies outside of $A$. We will show that it is sufficient to intersect $B$ with the boundary of $A$. To do so we need to ensure that Bregman balls are connected -- even if the domain $\Omega$ is restricted in an arbitrary way. This is important, since we want to handle arbitrary Bregman divergences over arbitrary (valid) domains, and not just $\mathbb{R}^d$ or its selected orthants.
\begin{maybeappendix}{legendreTransform}
    \myparagraph{Legendre transform.}
    The Legendre transform is a tool from convex geometry~\cite{Rockafellar+1970}. In the context of Bregman divergences, it is used to map a Bregman generator over a domain into another generator over a possibly different domain~\cite{Bregman_Voronoi}. In particular, it transforms primal balls in one domain into dual balls in the other domain, and vice versa. We will see one basic application of this tool.
    
    More technically, given a function of Legendre type, $F:\Omega\to \mathbb{R}$, there exists the \new{Legendre transform} which maps $F$ to a conjugate $F^*:\Omega^*\to \mathbb{R}$, where $\Omega^* = \{\nabla F(x)\,:\,x\in \Omega\}$ is the conjugate domain. Under this transformation, $F^*$ is also a function of Legendre type~\cite{Rockafellar+1970} and we can define the Bregman divergence associated to $F^*$. We now use this result to prove the connectedness of Bregman balls.
\end{maybeappendix}

\begin{lem}[Connectedness\ifappendix, Proof in Appendix~\ref{pf:connectedness}\fi]\label{lem:connectedness}
    Primal and dual Bregman balls are connected.
\end{lem}
\begin{maybeappendix}{pf2}
    \begin{proof}
        The dual balls are trivially convex~\cite{Bregman_Voronoi}, hence connected.

        %First consider the dual Bregman ball. Let $x,y\in B'_F(q;r)$ and $\ell$ be the line segement in $\Omega$ connecting the two. Then the points $(x,F(x))$ and $(y, F(y))$ lie below the vertically shifted tangent plane at $(q, F(q))$, and as the function $F$ is convex, we have that the lift of $\ell$ lies below this secant line of $(x,F(x))$ and $(y, F(y))$ by Jensen's Inequality~\cite{Rockafellar+1970}. As the lift of $x$ and $y$ both must lie on or below the shifted tangent plane, the secant line must lie on or below the plane. Therefore the lift of $\ell$ must also lie below the plane as well. So $B'_F(q;r)$ is convex and thus connected.
        The primal balls are more interesting, so we show an explicit proof. First, recall that $F$ is strictly convex and differentiable, implying it is continuously differentiable.
        Therefore, the Legendre transform of $F$ induces a homeomorphism $h:\Omega\to \Omega^*$. In particular, $h$ maps dual balls in $\Omega^*$ to primal balls in $\Omega$. Since connectedness is a topological property, any primal ball in $\Omega$ is connected as the homeomorphic preimage of a connected dual ball in $\Omega^*$.
    \end{proof}
\end{maybeappendix}

\ifappendix
While this result can be obtained from existing results~\cite{Bregman_Voronoi, EdWa16}in Appendix~\ref{pf:connectedness}, we provide a proof that accentuates the importance of the assumptions of Legendre-type functions, particularly assumption (III). These subtle details are often omitted from Bregman geometry literature, prompting us to provide detailed proofs.
\else
This particular proof is useful for clarifying the importance of the three assumptions in the definition of the Legendre-type function. (I) and (II) gives \emph{continuous} differentiability, and consequently the crucial homeomorphism. As for (III), let us show how things can go wrong without it. Specifically, if we allowed arbitrary convex restrictions of the domain. Consider $\Omega'$ as a restriction of $\Omega$ to the preimage under $h$ of a non-convex primal ball in $\Omega^*$. Since $\Omega'$ is convex, everything appears to work. However, the restricted $h$ now maps $\Omega'$ to a non-convex conjugate domain, where the Legendre transform is not well defined. The above proof would therefore fail if we restricted the domain in this way -- and we could not rule out the existence of non-connected balls. Condition (III) prevents us from making this mistake. Rockefellar~\cite{Rockafellar+1970} mentions that this is a very common mistake in general ---  it is also present in the Bregman divergence literature.
\fi

We now use connectedness for the following lemma.
\begin{lem}[Boundary Intersection\ifappendix, Proof in Appendix~\ref{pf:pruning}\fi]\label{lem:pruning}
Let $A \subset \mathbb{R}^d$ be an axis-aligned box of positive finite volume with boundary $\partial A$; $q \in \mathbb{R}^d \setminus A$ be the center of a Bregman ball $B$ of finite radius $r$.
    If  $B \cap \partial{A} = \emptyset$, then $A \cap \Omega$ lies in $\Omega \setminus B$.
\end{lem}
\begin{maybeappendix}{pf3}
    \begin{proof}
        Being a codimension-1 topological sphere, $\partial A$ divides $\mathbb{R}^d$ into the inside and the outside, by the Jordan--Brouwer separation theorem. Because $B \cap \partial A = \emptyset$, $q \notin A$, and $B$ is connected, we have that $B$ is necessarily on the outside of $\partial A$ and so $B \cap A = \emptyset$. Therefore, $A \cap \Omega$ indeed lies in the complement of $B$ in $\Omega$.
    \end{proof}
    \end{maybeappendix}
Thus $B\cap \partial A = \emptyset$ implies that the divergence from $q$ to each potential \emph{data} point in $A$ exceeds the radius of $B$, namely $r$. This means that the intersection test with the boundary of $A$ is sufficient to safely prune points in the Kd-tree query. We omit the analogous case for dual balls. We mention that the finite volume assumption is just a technicality as in practice Kd-trees partition a box of finite volume.

%See Figure~\ref{fig:GKL_v_Euc} for an illustration. 

\section{Proof of pruning efficiency}\label{sec:eff}
In this section we show that the pruning test can be performed efficiently in the case of decomposable Bregman divergences.
Specifically, we aim to update the projection in $O(1)$ running time, independently of the dimension of the ambient space. To this aim, we rephrase the pruning test in terms of a Bregman projection onto the \emph{boundary} of the box. 
We call this the \new{boundary projection test}.

\begin{lem}\ifappendix[Proof in Appendix~\ref{pf:domain_expansion}]\fi\label{lem:domain_expansion}
    Let $D'_F$ be a decomposable Bregman divergence defined on $\Omega'\times \Omega'$. Then $D'_F$ is a restriction of a divergence defined on an axis-aligned box\ifappendix, which we write as $\Omega=\prod_{i=1}^d\omega_i$\fi.
\end{lem}
\begin{maybeappendix}{pf4}
    \begin{proof}
        Let $D'_F$ be a decomposable Bregman divergence. Then $F= \sum_{i=1}^df_i$, where each $f_i$ is a univariate function of Legendre type. Thus, each $f_i$ has a convex domain in $\mathbb{R}$, $\omega_i$, which must be a interval. Thus, $D_F$ is a Bregman divergence defined on an axis-aligned box. Then, as both are generated by the same function of Legendre type, $D_F|_{\Omega'\times\Omega'} = D'_F$ and thus $D'_F$ is a restriction of $D_F$.
    \end{proof}
\end{maybeappendix}
This lemma,  further emphasizes the importance of Legendre-type function assumption (III), as a restriction on a Bregman divergence's domain is induced by a domain restriction on the parametrization function $F$.
 
 For a given query and set of data points the nearest neighbours are identical under both $D_F$ and a restricted $D'_F$, allowing use of either divergence. The assumption on the domain $\Omega = \prod_{i=1}^{d} \omega_i$ ensures that \cref{lem:axis-aligned-proj} and \cref{lem:box-proj} apply to Kd-trees. This is important because Kd-trees decompose $\mathbb{R}^d$ and not the chosen domain of a Bregman divergence. Additionally, in the unrestricted domain, \cref{lem:increment-projection} enables efficient query processing and ensures that our underlying algorithm remains robust under any domain restriction.
%We stress an important technicality: any decomposable divergence $D_F = \sum_{i=1}^{d}D_{f_i}$ is a restriction of a divergence defined on a box, namely $\prod_{i=1}^d \omega_i$, where each $\omega_i$ is the domain of $f_i$.
%When restricting a Bregman divergence, the induced restriction on the parametrization function $F$. This must still be a function of Legendre type, emphasizing property (III) and making restrictions nontrivial.

%As an example, let us consider the KL divergence, defined on the open standard simplex, $\triangle^{d-1} \subset \mathbb{R}^d$. It is a restriction of GKL on $\mathbb{R}_{+}^d$. To handle queries for data in $\triangle^{d-1}$ with KL, Kd-trees perform projections in $\mathbb{R}_{+}^d$ using GKL. More generally, we end up with exact Kd-tree queries: in the worst case, we may fail to prune points that would be pruned if we performed the projections inside the restricted domain (e.g. a simplex). However the divergence computations between data points in the restricted domain are not affected by these considerations. We remark that performing efficient projections inside the restricted domain would be much harder. It is in fact an interesting open problem, even for the special case of KL.
%Consider a decomposable Bregman divergence $D_{F}:\Omega\times\Omega\to \mathbb{R}$ and a restricted version $D'_F:\Omega'\times\Omega'\to\mathbb{R}$, where $\Omega'\subseteq\Omega$. Then for any $p,q\in \Omega'$ it holds that $D_F(p\|q) = D'_F(p\|q).$ Therefore, 

\myparagraph{From boxes to hyperplanes.}
We first consider a simplified problem, namely a Bregman projection onto a single axis-aligned hyperplane.% We return to the original problem afterwards. For the squared Euclidean distance, projecting onto an axis-aligned hyperplane is straightforward. We now consider a version for other decomposable Bregman divergences.

% %\begin{figure}
%     \centering
%     \includegraphics[width=0.45\textwidth]{figs/GKL_Plane.pdf}
%     \includegraphics[width=0.45\textwidth]{figs/IS_Euc_NN.pdf}
%     \caption{\figb{Left:} A GKL ball intersects the hyperplane and contains the ball of radius $\inf_{p\in P}D_{GKL}(q\|p)$. \figb{Right:} An IS ball intersects the hyperplane while the Euclidean ball does not, resulting in different nearest neighbours.}
%    \label{fig:GKL_v_Euc}
%\end{figure}

\begin{lem}[Axis-Aligned Projection\ifappendix, Proof in Appendix~\ref{pf:axis-aligned-proj}\fi]\label{lem:axis-aligned-proj}
Let $F=\sum_{i=1}^{d}f_i$ be a decomposable function of Legendre type defined on an axis-aligned box.%, where each $f_i$ has domain $\omega_i\subseteq\mathbb{R}$.

Let $P\subset\mathbb{R}^{d}$ be an axis-aligned hyperplane such that $P\cap\Omega\neq\emptyset$. Let $q_P$ be the Bregman projection of a point $q \in \Omega$ onto $P$  with respect to $D_F$. We claim that $q_P$ coincides with the orthogonal projection of $q$ onto $P$ for computation in either direction.
% and let $F$ be a separable function of Legendre type. Then the decomposable Bregman divergence generated by $F$ has a projection in $q$ 
\end{lem}
\begin{maybeappendix}{pf5}
\begin{proof}
Let $P$ be an axis-aligned hyperplane orthogonal to the $j$-th standard basis vector. Specifically, each point $p \in P$ has its $j$-th coordinate fixed.

%construct candidate
Because: (1) $q$ is fixed; (2) $p_j$ is fixed for each $p\in P$; and (3) $D_F$ is decomposable, we can write $D_{F}(q\|p) = \sum_{i=1}^{d}D_{f_i}(q_i \| p_i) = D_{f_j}(q_j\|p_j) + \sum_{i\neq j}D_{f_i}(p_i\|q_i)$. To minimize $D_F(q\|p)$, we minimize $\sum_{i\neq j}D_{f_i}(q_i\|p_i)$, and since each $D_{f_i}$ is a Bregman divergence, Bregman Nonnegativity (Property~\ref{positivity}) applies. So each $D_{f_i}(q_i\|p_i)\geq0$, with equality if and only if $p_i=q_i$.
Consequently, $\arg\inf_{p\in P \cap \Omega}D_{F}(q\|p) = (q_1,q_2,\dots,p_j,\dots,q_d)$. Therefore, the Bregman projection of $q$ onto $P$ is precisely the orthogonal projection.
\end{proof}
\end{maybeappendix}
Generally our Bregman projection \emph{from} $q$ would not be considered a Bregman projection, and Lemma \ref{lem:proj}) would not apply to it. In our case, the two projections coincide, so we will refer to the resulting point as the Bregman projection onto the axis-aligned hyperplane. With this we get the following corollary.

\begin{cor}[Box Projection Divergence]\label{lem:box-proj}
    %Under the same assumptions as in \cref{lem:axis-aligned-proj}, 
    Let $F=\sum_{i=1}^{d}f_i$ be a decomposable function of Legendre type defined on an axis-aligned box. The Bregman projection divergence of $q\in\Omega$ onto $A$, with respect to the Bregman divergence generated by $F$, can be computed as
\begin{align*}    
    \sum_{i=1}^d D_{f_i}(q_i \| p_i) = D_F(q \| p),
\end{align*}
where $p$ is the (squared) Euclidean projection of $q$ onto $A$.
\end{cor}

\myparagraph{Back to the pruning test.} To decide if the input points in the current box can be pruned, we compare two values. One is the divergence to the current best candidate; the second one is the projection divergence of $q$ onto an axis-aligned box $A$. This also works for divergence computed in the reverse direction.

In the end, the situation is very simple. This simplicity allows us to compute the projection divergence in time $O(d)$ -- exactly as in the (squared) Euclidean case. 

\myparagraph{Efficient projection.} We focus on \emph{maintaining} the projection divergence during the course of the query, rather than computing it every time. It turns out a single update can be done in constant time, independent of the dimension.

\begin{lem}[Updating Projection Divergence in Constant Time\ifappendix, Proof in Appendix~\ref{pf:increment-projection}\fi]\label{lem:increment-projection}
Let $F=\sum_{i=1}^df_i$ be a decomposable function of Legendre type, where each $f_i$ has domain $\omega_i\subseteq \mathbb{R}$. Then the projection divergence can be updated in constant time.
\end{lem}
\begin{maybeappendix}{pf6}
    \begin{proof}
        %We adapt an algorithmic trick used in ANN's Kd-tree query that works for the squared Euclidean case. It is based on two observations: (1) that the computation of the squared Euclidean distance decomposes by dimension; (2) that moving from one box to its immediate child only shifts one wall. This is exploited in the query algorithm: when moving from a splitting node to its unvisited child, the squared Euclidean distance needs to be updated along only a single dimension. This is illustrated in Figure~\ref{fig:div_proj_update_scheme}. Crucially, this allows for an O(1) update of the projection distance onto the new box. 
        
        Let $q\in \Omega$ and $B=\prod_{i=1}^d[a_i, b_i]$ be a box corresponding to a splitting node of our Kd-tree. By corollary 1, the Bregman projection of $q$ onto $B$ is on the boundary of $B$. Denote this Bregman projection $x = \arg\inf_{p\in B}D_F(q\| p)$. As $x$ lies on the boundary, $x_i$ is either $a_i, q_i,$ or $b_i$.
        
        For the box $C$ corresponding to a child node, we change only one wall of $B$ by the construction of the Kd-tree. Without loss of generality, $C = [c, b_1]\times\prod_{i=2}^d[a_i, b_i]$, with $a_1 < c < b_1$. For $y=\arg\inf_{p\in C}D_F(q\|p)$ we similarly have $y_i = x_i$ for $i=2,\dots,d$.
        
        Since $D_F$ is decomposable, $D_{f_i}(q_i\|x_i) = D_{f_i}(q_i\|y_i)$ for $i=2,\dots,d$. Thus $D_{F}(q\|\rho) = D_{F}(q\|\omega) - D_{f_1}(q_1\|\omega_1) + D_{f_1}(q_1\|\rho_1).$ This is illustrated in Figure~\ref{fig:div_proj_update_scheme}.
        
        Thus, as we move from a splitting node to its child, updating the projection divergence is independent of the dimension, $d$. 
    \end{proof}
\end{maybeappendix}
Moving from a Kd-tree node to its child, the corresponding box shrinks along a single dimension. The projection divergence can therefore be updated using at most two divergence computations (one negative, one positive) along the same dimension. The update is $O(1)$ and independent of the embedding dimension. See Figure~\ref{fig:div_proj_update_scheme}.

\begin{figure}
    \centering
    \includegraphics[width = .4\textwidth]{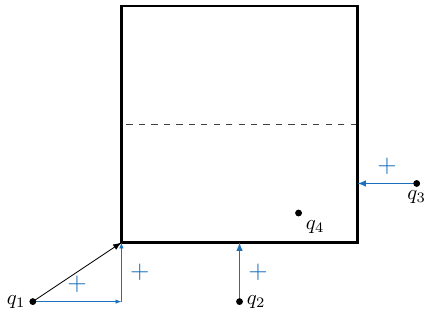}
    \includegraphics[width = .4\textwidth]{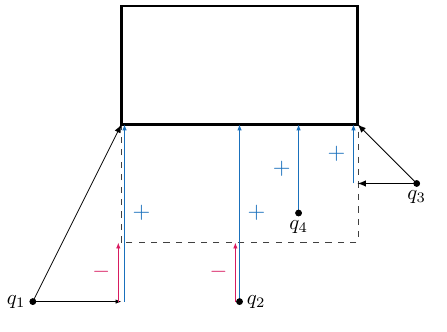}
    \caption{\figb{Left}: Calculation of projection divergences of each point $q_i$ onto the box decomposed as the sum of divergence computations along individual dimension. \figb{Right:} Efficient update of projection divergence.}
    \label{fig:div_proj_update_scheme}
\end{figure} 
As \cref{lem:box-proj} follows from \cref{lem:axis-aligned-proj} and \cref{lem:increment-projection} solely depends on the decomposable structure, the results apply for Bregman divergences computed in either direction.

\section{Implementation} \label{sec:implementation}
Our implementation is based on the ANN (Approximate Nearest Neighbour) C++ library by Mount and Arya~\cite{ANN_Manual}. It is an optimized library for Kd-trees. Extending the query algorithms to the Bregman setting requires significant modifications --- we therefore stress it is not a matter of simply using the library.

\myparagraph{Bregman query implementation.}
%In this case, we maintain the projection \emph{divergence}. From~\cref{lem:box-proj} (\nameref{lem:box-proj}) we know that the projection point onto an axis-aligned box does not depend on the choice of the decomposable Bregman divergence. We therefore reuse the algorithm, only changing the squared Euclidean distance computation, along a single dimension, to a decomposable Bregman divergence computation, along the same single dimension.
Algorithm~\ref{alg:ann} shows a C++ implementation of the Bregman query algorithm using this optimization (modulo unimportant technicalities). The code is structured after the implementation in the ANN library. We show only the part of the code for splitting nodes; handling leaf nodes is straightforward. 

We assume that splitting nodes are instances of class \texttt{kd\_tree\_splitting\_node}. Leaf nodes store input points and queries are handled with a linear search algorithm. Variable \texttt{eps} is used for approximate queries. Finally, \texttt{D\_f} is assumed to compute the decomposable Bregman divergence \emph{along a selected coordinate} -- for all practical decomposable divergences this takes time $O(1)$ by utilizing \cref{lem:increment-projection} at lines 18 and 20. Line 18 adds the new projection divergence, \texttt{new\_proj\_div}, and the old projection divergence is removed in line 20. We remark that many implementations, including KDTree from the popular sklearn library~\cite{scikit-learn}, use a slow $O(d)$ approach, adding unnecessary work in higher dimensions.

The variable \texttt{knn\_priority\_queue} is used to maintain the $k$-nearest neighbours.
To perform the dual query, just swap the parameters in the function used to compute \texttt{D\_f}. The extra argument in \texttt{D\_f} is just a technicality which allows one to use different 1-dimensional divergence depending on the currently considered dimension.

%We also implemented an analogous algorithm for the version using priority search, which reorders the visited nodes based on the projection distance. The Bregman part is analogous, so we are not showing it here.

\begin{algorithm}
\caption{Bregman Kd-tree query implementation}
\label{alg:ann}
\begin{lstlisting}[language=C++, escapechar = |]
struct kd_tree_splitting_node : kd_tree_node {
  kd_tree_node *child_lower, *child_higher;  
  int cut_dim;  
  float cut_val, upper_bound, lower_bound;
  virtual void search(...);
};

using div_t = std::function<float(const float, const float, const int)>;

float D_GKL(const float x_i, const float y_i, const int dim) { // example
  return x_i*log(x_i) - x_i*log(y_i) - x_i + y_i;
}

void virtual kd_tree_splitting_node::search(const point& q, 
               float box_proj_div, div_t D_f, float eps=0.0) {
  if (q[cut_dim] < cut_val) { // q lower than the cutting plane
    child_lower->search(q, box_proj_div, D_f, eps); // more promising child	
    float new_box_proj_div = box_proj_div+D_f(q[cut_dim], cut_val, cut_dim);
    if (lower_bound > q[cut_dim]) 
      new_box_proj_div -= D_f(q[cut_dim], lower_bound, cut_dim);
    if (box_div*(1+eps) < knn_priority_queue->max_divergence())
      child_higher->search(q, new_box_proj_div, D_f, eps); // recursive call
  }
  else {/* analogous for q higher than the cutting plane... */ }    
}
\end{lstlisting}
\end{algorithm}

\myparagraph{Expected computational complexity of a query.}
While we perform a single visit of an internal node in an optimal (constant) time, 
the expected complexity of the query remains an open problem. In particular,
proving the $O(\log n)$ bound for uniformly distributed data is significantly harder
in our setting. First, it is not clear what it means to uniformly distribute points
with respect to a given divergence. Second, standard proofs rely on volume arguments ---
but in our case the volume of a ball depends on its location. These issues necessitate new, significantly more sophisticated, proof techniques. We leave this as an open problem, and
show experimentally that the method performs well in practical situations.

\section{Experiments}
%like to utilize the usual computational geometry tools and ideas in the space...
The main practical motivation of our work is to play the usual computational geometry games in the point clouds produced by machine learning models. In particular, we wish to compare two collections of probabilistic (soft) predictions using the KL divergence. Efficient nearest neighbour queries are useful in this setting, however using the Euclidean tools leads to severe discrepancies. Finally, the dimension is often not overly high (often 10-100) which gives hope that Kd-trees can be efficient.

We will benchmark Bregman Kd-trees for exact and approximate nearest neighbour queries and compare with other methods. Additional results are reported in~\cref{sec:extra_tests}. 

We stress that efficiency is just one aspect determining the practicality of a method. Bregman Kd-trees have several unique advantages which make them practical. For example, once a Bregman Kd-tree is constructed, each query can be performed using a different decomposable Bregman divergence.

\myparagraph{Data sets.} 
We use synthetic data, standard datasets, and probabilistic predictions coming from a machine learning model. For all data, we use 50,000 points and 10,000 query points.

In the machine learning setup, we consider popular image datasets CIFAR10 and CIFAR100. Each contains 50,000 training and 10,000 test images, with 10 and 100 different labels respectively. We train two neural networks, $M_1$ and $M_2$, on a classification task on CIFAR100. % Specifically, we perform transfer learning using EfficientNetB0~\cite{tan2019efficientnet} pretrained on imagenet as a backbone, with ($M_1$) and without ($M_2$) fine-tuning. 
They achieve 80.22\%, and 71.74\% test accuracy respectively. From each model, we produce two sets of probabilistic predictions: $(\text{trn}_{i},\text{tst}_i)$, for $i \in \{1,2\}$. By $Q\to D$ we mean we query dataset $D$ with queries $Q$. Since the network is trained to minimize the total KL divergence, these predictions lie on the $\triangle^{99} \subset \mathbb{R}^{100}$ equipped with KL divergence.

We also consider a model trained to 95.2\% test accuracy on CIFAR10 and extract its probabilistic predictions on training and test points. These predictions are contained in $\mathbb{R}^{10}$. We also use the standard Corel Image Features data contained in $\triangle^{99}$. % and a uniform random of $\triangle^{99}$. 

\myparagraph{Compiler and hardware.}
Software was compiled with {Clang 14.0.3}. The experiments were done on a single core of a 3.5 GHz ARM64-based CPU with 4MB L2 cache using 32GB RAM. We observed similar speed ups on an x86-64 CPU.

\subsection{Nearest neighbour comparisons}
We first compare the sets of nearest neighbours obtained by using Euclidean distance and the KL divergence.
For tst$_1\to$trn$_1$, we find the 10 nearest for each query in tst$_1$ with respect to the KL divergence and Euclidean distance.
%For 10,000 queries in tst$_{1}$, we find the 10 nearest neighbours for each query in trn$_{1}$ with respect to the KL divergence and Euclidean distance. 
For each query, we compare the two sets of neighbours while disregarding order. Of the 10,000 queries, 9,962 had different sets of nearest neighbours with 134 having no common nearest neighbours. This is expected as the geometry of KL balls can vary depending on the location of the center whereas the Euclidean balls grow more uniformly. A lower dimensional example with three nearest neighbours computed can be seen in Figure~\ref{fig:disjoint_nns}. When the sets are ordered, the average of number of neighbours with matching indices is 0.8422, with only two queries having the same nearest neighbours in the same order.
\begin{figure}
    \centering
    \includegraphics[width = 0.4\textwidth]{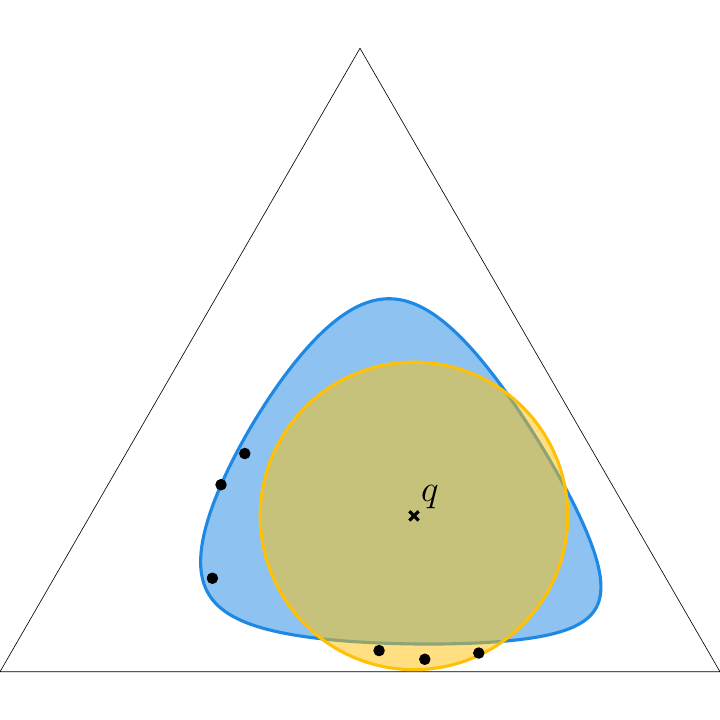}
    \caption{Three nearest neighbours of a query $q$ with respect to the KL divergence and the Euclidean distance on $\triangle^{2}$. The blue area is the KL ball and the yellow is the Euclidean ball whose radius is determined by the respective third nearest neighbour.}
    \label{fig:disjoint_nns}
\end{figure}

\subsection{Baselines} We stress that there are no robust libraries for the exact and approximate Bregman nearest neighbour computations. We compare our package to Cayton's experimental implementation of Bregman ball-trees~\cite{CaytonBBTrees} (\new{BBT}) for exact and approximate nearest neighbours. Two other available (experimental) implementations~\cite{VPTree} are not usable, due to severe compilation issues (the code is non-portable) and limited documentation.

We additionally compare our package to the fastest implementation in NMSLIB~\cite{nmslib} for Bregman divergences. We stress that these  methods are recall-approximate: they may return the correct nearest neighbour, but offer no guarantees (unlike our method). These methods are therefore not in direct competition with our method, but it is interesting to observe the trade-offs between efficiency and guarantees.
%More generally, while there is a vast literature, implementations and benchmarks on non-metric similarity search (for example using the cosine dissimilarity) -- Bregman divergences are not included. 

\subsection{Exact queries}
In Table~\ref{tab:tree_build_time}, we measure the construction time of Kd-trees and ball trees. We note that the same Kd-tree works for any decomposable divergence for either direction, while a ball tree depends on the given divergence and direction.

\begin{table}
    \centering
    \caption{Build time for Kd-trees and Bregman ball-trees for different divergences and data sets.}
    \begin{tabular}{l r r r r}\toprule
        & Kd-tree & SqEuc ball tree & KL ball tree & IS ball tree \\\cmidrule(lr){2-2}\cmidrule(lr){3-3}\cmidrule(lr){4-4}\cmidrule(lr){5-5}
         trn$_1$  & 0.43s & 5.75s & 10.05s & 7.71s\\
         Corel 64 & 0.08s & 1.30s & 6.27s  & 6.08s\\
         CIFAR10  & 0.04s & 0.20s & 1.30s  & 0.88s\\
         \bottomrule
    \end{tabular}
    \label{tab:tree_build_time}
\end{table}
Table~\ref{tab:Exact_Comps} shows the speed up in finding nearest neighbours using our method compared to linear search and BBT. We use the KL, IS, BL, and a hybrid divergence. For exact queries with the KL divergence on the 100-dimensional CIFAR data sets we observe $\approx100\times$ speed compared to the linear search. We compare our speeds to Cayton's Bregman ball trees, achieving minimum $3\times$ speed up with KL divergence and up to $20\times$ speed up for the IS divergence. As BBT has not implemented BL or hybrid divergences, these times are not available.

\begin{table}
    \centering
    \caption{Runtimes of Kd-trees compared to Bregman ball-trees and linear search. Speed ups are in comparison to Kd-tree times. The Hybrid divergence is 0.9KL + 0.1Euc.}

\begin{tabular}{llrrrrrrrr}\toprule
                                                                                 &        & \multicolumn{2}{c}{tst$_1\to$trn$_1$} & \multicolumn{2}{c}{tst$_2\to$trn$_1$} & \multicolumn{2}{c}{Corel64} & \multicolumn{2}{c}{CIFAR10} \\\cmidrule(lr){3-4}\cmidrule(lr){5-6}\cmidrule(lr){7-8}\cmidrule(lr){9-10}
\multirow{4}{*}{\rotatebox{90}{Kd-tree}}                                                         & KL     & \multicolumn{2}{c}{3.06s}             & \multicolumn{2}{c}{3.66s}             & \multicolumn{2}{c}{18.26s}  & \multicolumn{2}{c}{0.30s}   \\
                                                                                 & IS     & \multicolumn{2}{c}{24.95s}            & \multicolumn{2}{c}{26.65s}            & \multicolumn{2}{c}{67.95s}  & \multicolumn{2}{c}{1.13s}   \\
                                                                                 & BL     & \multicolumn{2}{c}{4.72s}             & \multicolumn{2}{c}{5.78s}             & \multicolumn{2}{c}{27.05s}  & \multicolumn{2}{c}{0.49s}   \\
                                                                                 & Hybrid & \multicolumn{2}{c}{5.55s}             & \multicolumn{2}{c}{6.44s}             & \multicolumn{2}{c}{21.94s}  & \multicolumn{2}{c}{0.50s}   \\\midrule\midrule
\multirow{4}{*}{\rotatebox{90}{Linear}} & KL     & 281.90s        & 92.12$\times$        & 286.63s        & 78.31$\times$        & 177.77s    & 9.74$\times$   & 30.53s   & 101.77$\times$   \\
                                                                                 & IS     & 277.87s        & 11.14$\times$        & 274.58s        & 10.30$\times$        & 173.79s    & 2.05$\times$   & 30.20s   & 23.01$\times$    \\
                                                                                 & BL     & 88.59s         & 18.77$\times$        & 87.84s         & 15.20$\times$        & 55.46s     & 2.05$\times$   & 8.89s    & 18.14$\times$    \\
                                                                                 & Hybrid & 309.21s        & 55.71$\times$        & 311.63s        & 48.39$\times$        & 196.99s    & 8.98$\times$   & 34.20s   & 68.40$\times$    \\\midrule\midrule
\multirow{4}{*}{\rotatebox{90}{BBT}}    & KL     & 9.62s          & 3.14$\times$         & 14.33s         & 3.92$\times$         & 99.10s     & 5.43$\times$   & 0.91s    & 3.03$\times$     \\
                                                                                 & IS     & 507.45s        & 20.34$\times$        & 614.97s        & 23.08$\times$        & 397.97s    & 5.86$\times$   & 3.53s    & 3.53$\times$     \\
                                                                                 & BL     & N/A            &                      & N/A            &                      & N/A        &                & N/A      &                  \\
                                                                                 & Hybrid & N/A            &                      & N/A            &                      & N/A        &                & N/A      &                 \\\bottomrule
\end{tabular}
    \label{tab:Exact_Comps}
\end{table}
\begin{maybeapprox}{ApproxExperiments}
\ifapprox
\subsection{Approximate query experiments}
\else
\subsection{Approximate Bregman queries}
\fi
Given $\epsilon$, an \new{approximate nearest neighbour query} must return each nearest neighbour $x'$ such that $D_F(q \| x') \le (1+\epsilon)D_F(q \| x)$, where $x$ is the true nearest neighbour.

To evaluate our method for approximate queries, we compare it with an implementation of Bregman Ball trees (\new{BB-trees}) by Cayton~\cite{CaytonBBTrees, BBTreeGithub}.  It is specialized for KL, with \emph{experimental support} for IS. Unlike our method, extending it to other divergences is nontrivial.

%One method for searching for approximate nearest neighbours in the Euclidean metric is to apply a $(1+\epsilon)^2$ buffer. For the case of Bregman divergences we instead use a $(1+\epsilon)$ factor for interpretability. This agrees with Cayton's ball tree implementation for approximate queries.

We compare the query times for the Kd-tree search to the Bregman ball tree search for a range of $\epsilon$ values in Figure~\ref{fig:kd_bb_eps}. Our method is between 3-5 times faster for KL queries, and between 5-15 times faster for IS queries.
\begin{figure}
    \centering
    \includegraphics[width = .49\textwidth]{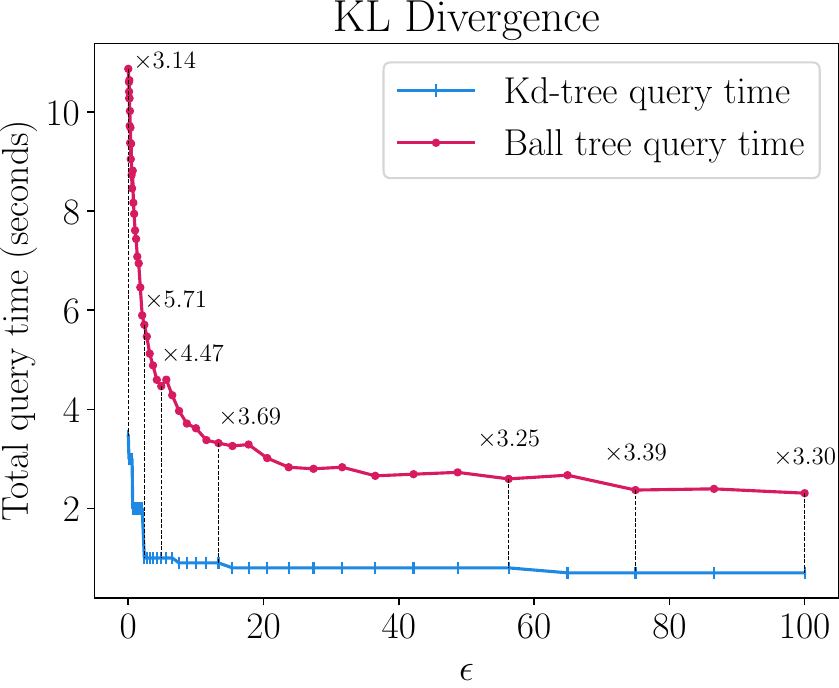}
    \includegraphics[width = .49\textwidth]{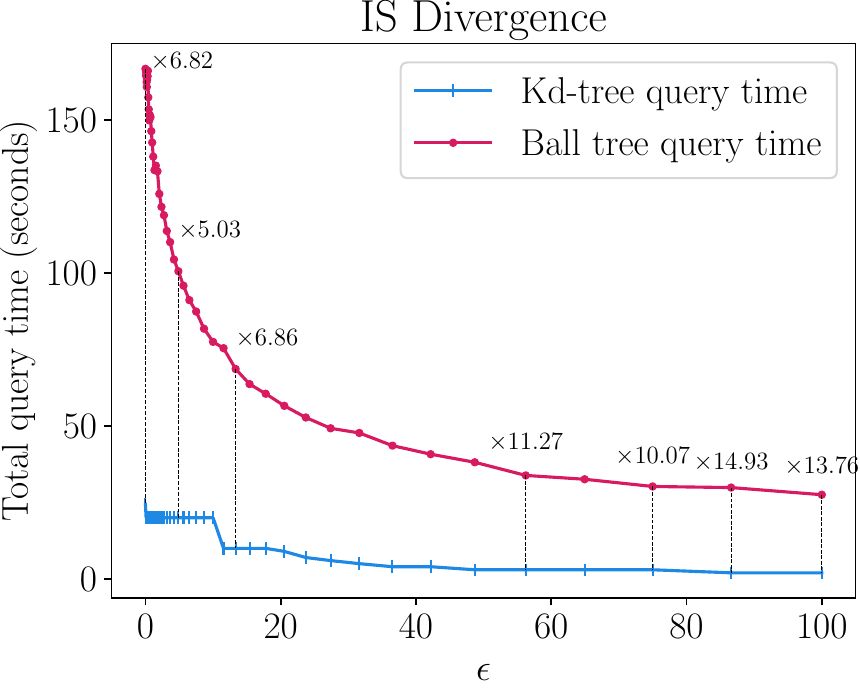}
    \caption{Total query time compared for $(1+\epsilon)$-approximate nearest neighbours for tst$_{1}\to$trn$_{2}$ (lower is better). Left is KL and right is IS divergence. Starting from $\epsilon = 0.1$. Vertical bars mark the speed up of Kd-trees over ball-trees for a given $\epsilon$.}
    \label{fig:kd_bb_eps}
\end{figure}

Our Kd-tree method works with arbitrary decomposable divergences (computed in either direction): one can either use a predefined divergence, or implement a custom one. This only requires implementing a single function in the user's code that computes the divergence -- no changes to the Kd-tree library are required. This allows the method to work out of the box in various contexts.

The above is in contrast with Bregman Ball trees: they are more general but require tailoring to different divergences~\cite{CaytonBBTrees}. Also, there is a big difference between the simple squared Euclidean case and the Bregman case. Finding a projection onto a (dual) Bregman ball generally requires performing a 1-dimensional convex optimization. In practice, this is done using a binary search, with a full divergence computation at each step, making each step $\Omega(d)$. In our case this entire projection is $O(1)$. As evidenced by BB-tree's relatively lower performance for the IS divergence, extending BB-trees to other divergences poses an algorithmic challenge.

\subsection{An unfair comparison with fast heuristics}
The Non-Metric Space Library~\cite{nmslib} by Naidan has algorithms adapted for working with Bregman divergences and other non-metric distances. Benchmarks for methods have been published specifically for the KL and IS divergences~\cite{nmslib_benchmarks}. In particular, the small-world graph (SWG)~\cite{hnswg_search} search is considered a state-of-the-art method. For brevity (and because these heuristic methods do not directly compete with methods for exact and approximate queries), we limit the comparison to this method. 

Unlike Kd-trees, the SWG method does not offer guarantees on the number of correct nearest neighbours. While they tend to behave well in practice,  Indyk and Xu showed~\cite{indyk23} that such methods can fail catastrophically. Generally, recall tends to drop in higher dimensions. In particular, in dimension 100, 15.71\% results for the KL divergence contained some incorrect nearest neighbours; in 4.54\% of cases, \emph{all} of the reported neighbours were incorrect. 

In any case, the benchmarks reported in Table~\ref{tab:Kd_SW_Comp} reveal an interesting trade-off. Compared to our implementation, the SWG offers faster query time (typically one order of magnitude faster), at the cost of slower build time (typically two orders of magnitude slower). We reiterate that these methods do not provide performance guarantees --- while our method does. Overall, SWG is useful for performing numerous imprecise searches, while Kd-trees are useful for
fewer searches or when guarantees are required.

\begin{table}
\centering
    \caption{Comparing Kd-tree and SWG method for 10 nearest neighbours. For SWG, error frequency is number of times $<10$ correct nearest neighbours are returned. Min recall is the minimum number of correct nearest neighbours. The Hybrid (H) divergence is 0.9KL+0.1Euc.}
    %\resizebox{\textwidth}{!}{%
\begin{tabular}{l l r r r r r r r r r}\toprule
 &  & \multicolumn{1}{c}{\shortstack{Kd-tree\\Build\\time}} & \multicolumn{1}{c}{\shortstack{Kd-tree\\query\\time}} & \multicolumn{1}{c}{\shortstack{SWG\\build\\time}} & \multicolumn{1}{c}{\shortstack{SWG\\query\\time}} & \multicolumn{1}{c}{\shortstack{Avg\\recall}} & \multicolumn{1}{c}{\shortstack{Error\\freq}} & \multicolumn{1}{c}{\shortstack{Min\\recall}} & \multicolumn{1}{c}{\shortstack{Min\\recall\\freq}} &  \\\cmidrule(lr){3-3}\cmidrule(lr){4-4}\cmidrule(lr){5-5}\cmidrule(lr){6-6}\cmidrule(lr){7-7}\cmidrule(lr){8-8}\cmidrule(lr){9-9}\cmidrule{10-10}
\multirow{3}{*}{KL} & tst$_2\to$trn$_1$ & 0.43s & 5.76s & 4.84s & 1.32s & 0.928 & 1571 & 0 & 454 \\
                    & Corel64           & 0.08s & 29.60s& 6.99s & 1.70s & 0.998 & 152  & 7 & 5\\
                    & CIFAR10           & 0.04s & 0.45s & 2.90s & 0.59s & 0.998 & 76   & 1 & 1 \\\midrule\midrule
\multirow{3}{*}{IS} & tst$_2\to$trn$_1$ & 0.43s & 29.29s& 8.84s & 2.31s & 0.883 & 3292 & 0 & 387 \\
                    & Corel64           & 0.08s & 76.03s& 17.13s& 4.05s & 0.900 & 4608 & 0 & 1 \\
                    & CIFAR10           & 0.04s & 2.05s & 2.90s & 0.85s & 0.991 & 287  & 0 & 17  \\\midrule\midrule
\multirow{1}{*}{BL} & All data sets     &       &       & N/A   & N/A   & N/A   & N/A  & N/A & N/A\\\midrule\midrule
\multirow{1}{*}{H} & All data sets&  &             & N/A   & N/A   & N/A   & N/A  & N/A & N/A\\\bottomrule
\end{tabular}%}
    \label{tab:Kd_SW_Comp}
\end{table}

\end{maybeapprox}

\section{Summary}
We proved several results on Bregman divergences, demonstrating that the geometries they induce are well-behaved. In particular, we show that the lack of symmetry and triangle inequality does not preclude them from
being used as a measurement for Kd-trees. This is perhaps unexpected, since the triangle inequality is typically used to prove the correctness of Kd-trees. Furthermore, we show that certain additional properties of decomposable Bregman divergences enable an efficient query algorithm. These theoretical results provide the basis for an efficient implementation, whose properties are outlined below.

\begin{itemize}
    \item Computational complexity: a crucial operation is optimized to work in $O(1)$ time. In comparison, several popular Euclidean Kd-trees implementations use a naive $O(d)$ algorithm.
    \item Speed: it is up to 100$\times$ faster than linear search and between 3 and 20$\times$ faster than competing methods on practical data in dimension $100$.    
    \item Simplicity: the algorithm is simple which makes it more likely to be adopted in practice.
    \item Ease of use: works for any decomposable Bregman divergence out of the box (competing approaches requires custom, nontrivial implementation for each divergence).
    \item Flexibility: handles exact and guaranteed $\epsilon$-approximate Bregman queries with divergence  computed in either direction.
\end{itemize}

From an applied perspective, one can now perform efficient queries for practical data measured with the KL divergence, in particular on medium-dimensional data coming from machine learning. This opens up new ways of using computational geometry algorithms within machine learning. 

On the theoretical side, this work opens up new questions. Of primary importance is the expected computational complexity of a Kd-tree query. This problem is  significantly more involved than in the Euclidean case, and will require developing novel proof techniques and deepening our understanding of the geometries induced by Bregman divergences.

%In particular, we can now efficiently query collections of probability vectors measured with the Kullback--Leibler divergence. One interesting source of such data are predictions output by artificial neural networks, which naturally use the KL divergence. This is likely related to the low intrinsic dimension of such data. Interestingly using another divergence for this data results in much slower computations -- suggesting that a meaningful definition of intrinsic dimension should depend on the considered distance or divergence.

%We stress that Kd-trees have some unique features that distinguish them from alternative methods, like ball-trees. In particular, we guarantee $O(1)$ worst-case complexity for each node visit (independent from the dimension of the data). Also, the construction phase is oblivious to the choice of divergence.
%It makes the method an interesting choice for situations in which queries with respect to a number of different divergences are performed.

%Overall, the main takeaway is that Bregman Kd-trees achieve surprisingly good speedups on practical data sets.

\bibliographystyle{plainurl}
{
% \footnotesize
\bibliography{main, cubicle}
}

\appendix
\section{Additional tests} \label{sec:extra_tests}
\subsection{Higher dimensional experiments}
%\myparagraph{Randomly sampling the simplex}
%For the following experiments, we have three methods for sampling an open simplex. The first is to uniformly sample $(0,1)^d$ and normalize by the $L_1$ norm, the second is to uniformly sample the unit circle and apply the map $(x_i)_{i=1}^d\mapsto (x_i^2)_{i=1}^d$, coming from the Fisher isometry. The third is due to Devroye~\cite{DevroyeSampling}. Visualizations of the three methods can be seen in Figure~\ref{fig:Simplex_samplings}. In the three images, there is a difference in where the points tend to cluster. The first method is sparse around the corners, the second is dense around the edges and corners, and the third is more classically uniform. In context of machine learning, we often see data points clustering at the corners of the simplex, best modeled by the second method.
% \begin{figure}
%     \centering
%     \includegraphics[width=0.3\linewidth]{figs/L1_Cube.pdf}
%     \includegraphics[width=0.3\linewidth]{figs/FisherSphere.pdf}
%     \includegraphics[width=0.3\linewidth]{figs/Devroye.pdf}
%     \caption{Distribution of the three methods of sampling $\triangle^2$. Left-most figure is the result of sampling (0,1)^3 with $L_1$ normalization, center figure is the result of sampling the $\mathbb{R}^3$ sphere's positive orthant and applying the Antonelli isometri, and the right-most figure is from Devroye's sampling method.}
%     \label{fig:Simplex_samplings}
% \end{figure}

\myparagraph{Baselines in higher dimensions.}
For 50,000 data points and 10,000 query points sampled sampled from the simplex, we can compare query times between Kd-trees and Cayton's BBT. In Table~\ref{tab:dimension_comparison}, we compare query times in higher dimensions. Although Kd-trees are often said to be slower in higher dimensions, we see a speed up in our method even at $500$ dimensions.
\begin{table}[H]
    \centering
    \caption{Kd-tree and BBT 10 nearest KL--neighbours search times for increasing dimensions.}
\begin{tabular}{@{}lllllll@{}}\toprule
\multicolumn{1}{c}{Dimension}     & \multicolumn{1}{c}{100}     & \multicolumn{1}{c}{150}     & \multicolumn{1}{c}{200}     & \multicolumn{1}{c}{250}     & \multicolumn{1}{c}{500}       & \multicolumn{1}{c}{1000}      \\\cmidrule(lr){1-1}\cmidrule(lr){2-2}\cmidrule(lr){3-3}\cmidrule(lr){4-4}\cmidrule(lr){5-5}\cmidrule(lr){6-6}\cmidrule(lr){7-7}
Kd-tree query & 208.39s & 321.54s & 469.21s & 588.67s & 1,229.72  & 2,689.70s \\
BBT query     & 247.80s & 402.93s & 517.07s & 625.38s & 1,237.84s & 2,410.28s \\
Speed-up      & $\approx1.19\times$ & $\approx1.25\times$ & $\approx1.10\times$ & $\approx1.06\times$ & $\approx1.00\times$ & $\approx0.90\times$\\\bottomrule
\end{tabular}
    \label{tab:dimension_comparison}
\end{table}

\myparagraph{SW-graph comparison}
In comparison to Kd-trees, SWG has slow build times. We compare the SWG build time on 500,000 points in $\triangle^{999}\subset \mathbb{R}^{1000}$, with parameters reduced for speed while maintaining >0.9 average recall for 10 queries. SWG build time was 1464.07s, while Kd-tree build time and query time were 19.44s and 25.76s respectively. The build time for SWG is >30$\times$ longer than the sum of build and query time for Kd-trees.

\subsection{Other exact query experiments}
For these additional tests, $\square^{100}$ is a uniform sample of the unit cube with the same data sizes as above. In \cref{tab:more_queries}, we record total query time of other possible pairs of prediction data and $\square^{100}.$
\begin{table}
    \centering
    \caption{Additional total query time comparisons.}
\begin{tabular}{l l r r r r}\toprule
&& \multicolumn{1}{c}{tst$_{1}\to$trn$_{2}$} & \multicolumn{1}{c}{$\square^{100}$}  &\multicolumn{1}{c}{trn$_{2}\to$trn$_{1}$} & \multicolumn{1}{c}{trn$_{1}\to$trn$_{2}$}\\\cmidrule(lr){3-3}\cmidrule(lr){4-4}\cmidrule(lr){5-5}\cmidrule(lr){6-6}
\multirow{3}{*}{KL}    & Kd-Tree       & 4.46            & 17.22         & 12.73          &  264.84s          \\
                       & Linear Search & 285.16          & 1435.50       & 1431.93        &  407.80s          \\ \cline{3-6}
                       & Speed up      & 63.94$\times$   & 83.36$\times$ & 112.49$\times$ &  $1.54\times$     \\ \midrule\midrule
\multirow{3}{*}{IS}    & Kd-Tree       & 26.76           & 121.49        & 180.55         &  170.58s          \\
                       & Linear Search & 277.77          & 1377.00       & 1365.40        &  397.75s          \\ \cline{3-6}
                       & Speed up      & 10.38$\times$   & 11.33$\times$ & 7.56$\times$   &  $2.33\times$     \\ \midrule\midrule
\multirow{3}{*}{SqEuc} & Kd-Tree       & 0.41            & 1.89          & 1.87           &  89.95s           \\
                       & Linear Search & 23.32           & 116.44        & 116.466        &  24.36s           \\ \cline{3-6}
                       & Speed up      & 56.88$\times$   & 61.61$\times$ & 62.28$\times$  &  $0.27\times$     \\\bottomrule
\end{tabular}
    \label{tab:more_queries}
\end{table}

\ifappendix
\section{Proofs}\label{sec:proofs}
Here we provide careful, explicit proofs of the theorems stated in the paper. While the theorems are simple, the proofs are often subtle, as stressed in~\cite{Rockafellar+1970} in a more general context. The proofs assume the divergence computed \emph{from} a query $q$; their counterparts for the opposite direction are analogous.
%% Moved from Introduction
%    \item Careful, explicit proofs of its correctness --- despite the lack of triangle inequality. (While the theorems are simple, such proofs are often subtle, as stressed in~\cite{Rockafellar+1970}.)

\subsection{Proof of \cref{lem:connectedness}}\label{pf:connectedness}
We start by introducing a useful concept.
\fi

\putmaybeappendix{legendreTransform}

\ifappendix
\subsection{Proof of \cref{lem:pruning}}\label{pf:pruning}
\begin{relem}[Boundary Intersection]
Let $A \subset \mathbb{R}^d$ be an axis-aligned box of positive finite volume with boundary $\partial A$; $q \in \mathbb{R}^d \setminus A$ be the center of a Bregman ball $B$ of finite radius $r$.
    If  $B \cap \partial{A} = \emptyset$, then $A \cap \Omega$ lies in $\Omega \setminus B$.
\end{relem}\fi
\putmaybeappendix{pf3}

\ifappendix
\subsection{Proof of \cref{lem:domain_expansion}}\label{pf:domain_expansion}
\begin{relem}
    Let $D'_F$ be a decomposable Bregman divergence defined on $\Omega'\times \Omega'$. Then $D'_F$ is a restriction of a divergence defined on an axis-aligned box.
\end{relem}
\fi
\putmaybeappendix{pf4}

\ifappendix
\subsection{Proof of \cref{lem:axis-aligned-proj}}\label{pf:axis-aligned-proj}
% Solution to Hyperplane Projection Problem
% Let D_F= \sum D_{f_i}, since each $f_i:\omega_i\to \R$ is legendre type, $D_F$ can be defined on $\left(\prod \omega\right)^2$?
\begin{relem}[Axis-Aligned Projection]
Let $F=\sum_{i=1}^{d}f_i$ be a decomposable function of Legendre type defined on an axis-aligned box.%, where each $f_i$ has domain $\omega_i\subseteq\mathbb{R}$.

Let $P\subset\mathbb{R}^{d}$ be an axis-aligned hyperplane such that $P\cap\Omega\neq\emptyset$. Let $q_P$ be the Bregman projection of a point $q \in \Omega$ onto $P$  with respect to $D_F$. We claim that $q_P$ coincides with the orthogonal projection of $q$ onto $P$ for computation in either direction.
% and let $F$ be a separable function of Legendre type. Then the decomposable Bregman divergence generated by $F$ has a projection in $q$ 
\end{relem}\fi
\putmaybeappendix{pf5}

\ifappendix
\subsection{Proof of \cref{lem:increment-projection}}\label{pf:increment-projection}
\begin{relem}[Updating Projection Divergence in Constant Time]
Let $F=\sum_{i=1}^df_i$ be a decomposable function of Legendre type, where each $f_i$ has domain $\omega_i\subseteq \mathbb{R}$. Then the projection divergence can be updated in constant time.
\end{relem}
\fi
\putmaybeappendix{pf6}

\ifapprox
\section{Approximate searches}
\subsection{Related work}
Many search methods have been adapted to work with Bregman divergences.
Ring-trees combined with a quad-tree decomposition have been proven to work sublinearly for finding approximate nearest neighbours by Abdullah, Moeller, and Venkatasubramanian~\cite{Bregman_ring_tree}. In 2013, Boytsov and Naidan developed their own Bregman VP-trees extension~\cite{BoytsovNaidan_VPTrees} for approximate nearest neighbours. Naidan later incorporated his VP-trees and Cayton's ball-trees into the Non-Metric Space Library (NMSLIB)\cite{nmslib}. This library also includes other approximate Bregman similarity searches including small world graphs~\cite{MalkovPonomarenkoLogvinovKrylov_SWG}. The hierarchical navigable small world graph has been a popular choice for similarity searches in vector databases~\cite{HNSW_Zilliz, HNSW_MariaDB, HNSW_MongoDB} and perform well in benchmarks for metrics~\cite{ANN_Benchmarks}. However, its implementation in NMSLIB is currently experimental for Bregman divergences~\cite{HNSW_Git}.
Recently Abdelkader, Arya, da Fonseca and Mount proposed an approach to proximity search in non-metric settings, which includes Bregman divergences~\cite{abdelkader2019approximate}; as we understand it, this has not yet been implemented.

\fi
\putmaybeapprox{ApproxExperiments}
\end{document}